\documentclass[11pt,a4paper]{article}
\usepackage[utf8]{inputenc}
\usepackage[T1]{fontenc}
\usepackage[english]{babel}
\usepackage{amsmath, amssymb, amsthm}
\usepackage{mathtools}
\usepackage{physics}
\usepackage{graphicx}
\usepackage{float}
\usepackage{booktabs}
\usepackage{caption}
\usepackage{microtype}
\usepackage[colorlinks=true, linkcolor=blue!70!black, citecolor=teal, urlcolor=blue!70!black]{hyperref}
\usepackage[margin=1in]{geometry}

\usepackage{authblk}

\usepackage{tikz}
\usetikzlibrary{automata, positioning, arrows.meta, decorations.pathmorphing, calc, decorations.markings, shadows}

\tikzset{
	state/.style={
		circle, 
		draw=black, 
		thick, 
		minimum size=1.3cm, 
		inner sep=0pt, 
		font=\small\bfseries,
		fill=white,
		drop shadow={opacity=0.1}
	},
	final/.style={
		double, 
		draw=black, 
		thick, 
		minimum size=1.3cm,
		fill=white
	},
	transition/.style={
		->, 
		>={Stealth[length=2.5mm]}, 
		thick, 
		draw=black!85
	},
	edge label/.style={
		midway,
		above,
		font=\footnotesize,
		color=black!90,
		inner sep=2pt
	}
}

\newtheorem{theorem}{Theorem}
\newtheorem{lemma}{Lemma}
\newtheorem{proposition}{Proposition}

\newtheorem{remark}{Remark}

\title{\textbf{Matrix bordering structure of the Faddeev-Jackiw algorithm: kernel reduction and symbolic automation}}

\author[1]{E. Chan--López}
\author[2]{A. Martín-Ruiz\thanks{Corresponding author: \texttt{alberto.martin@nucleares.unam.mx}}}
\author[1]{Jaime Manuel Cabrera}
\author[1]{Jorge Mauricio Paulin Fuentes}

\affil[1]{División Académica de Ciencias Básicas, Universidad Juárez Autónoma de Tabasco (UJAT), Villahermosa, Tabasco, México}
\affil[2]{Instituto de Ciencias Nucleares, Universidad Nacional Autónoma de México (ICN-UNAM), Ciudad de México, México}

\date{\today}

\begin{document}
	
	\maketitle

	\begin{abstract}
		\noindent
		We show that the iterative Faddeev-Jackiw (FJ) reduction for singular Lagrangian systems constitutes a \textit{geometrically constrained instance} of the Matrix Bordering Technique (MBT). For a first-order Lagrangian with singular pre-symplectic form, each iteration of the Barcelos-Neto-Wotzasek algorithm produces an extended symplectic matrix of canonical bordered form,
		\[
		f^{(m)} =
		\begin{pmatrix}
			f^{(0)} & B \\
			- B^\top & 0
		\end{pmatrix},
		\]
		where the bordering block $B$ is determined by the gradients of the consistency constraints. Because $f^{(0)}$ is singular by hypothesis, the classical Schur determinantal identity---which requires an invertible anchor block---cannot be invoked. The regularity of $f^{(m)}$ is instead governed entirely by the interaction of $B$ with $\ker(f^{(0)})$. Writing $N$ for an orthonormal basis of $\ker(f^{(0)})$, we introduce the \emph{reduced constraint matrix} $\Gamma = N^{\top} B$ and prove the exact factorization $\det(f^{(m)}) = \bigl(\prod_\ell \mu_\ell^2\bigr)\,\det(\Gamma)^2$, where the $\mu_\ell>0$ are the Pfaffian factors of the nonsingular block of $f^{(0)}$. This yields the sharp nondegeneracy criterion: $f^{(m)}$ is regular if and only if the number of generated constraints equals $\dim\ker(f^{(0)})$ and $\Gamma$ is nonsingular. Our central algebraic result is that this reduced matrix coincides---coefficient by coefficient---with the Hessian of the symplectic potential restricted to $\ker(f^{(0)})$,
		\[
		\Gamma_{\alpha\beta} \;=\; N_a^{\,i}\,\partial_i\partial_j V\,N_b^{\,j} \;=\; \{\Omega_\alpha, \Omega_\beta\}_{\mathrm{FJ}},
		\]
		which is precisely the constraint matrix whose nondegeneracy characterizes a second-class system in the Dirac--Bergmann sense.
		As a consequence, we prove the biconditional $\det(f^{(m)}) \neq 0 \iff \det(\Gamma) \neq 0$: assuming the generated constraints are independent, the Faddeev-Jackiw algorithm terminates in a nondegenerate symplectic form if and only if the constraint algebra is nondegenerate, i.e., when the constraints form a second-class system.
		
		This algebraic characterization provides a rigorous foundation for automating the Faddeev–Jackiw procedure symbolically. We present a fully symbolic implementation in the Wolfram Language, and validate the approach on representative mechanical systems with nontrivial constraint structure. The resulting rule-based engine preserves parametric dependencies throughout the reduction, enabling reliable analysis of degeneracy, structural stability (when no bifurcations occur), and possible bifurcation scenarios as critical parameters are varied.
	\end{abstract}
	
	

	\section{Introduction}\label{sec1}

	Singular Lagrangian systems, those whose Hessian with respect to velocities is degenerate, lie at the heart of modern theoretical physics, appearing in gauge theories~\cite{Henneaux1992,Sundermeyer1982}, general relativity~\cite{Arnowitt1962,DeWitt1967}, and topological models~\cite{Witten1989,Schwarz1993}. Their consistent Hamiltonian formulation requires careful treatment of constraints, a problem historically addressed by Dirac's canonical formalism~\cite{Dirac1964}, which classifies constraints into first- and second-class and introduces secondary conditions through consistency requirements.

	In 1988, Faddeev and Jackiw proposed an elegant geometric alternative~\cite{FaddeevJackiw1988} that reformulates constrained dynamics in terms of a pre-symplectic two-form, avoiding the explicit hierarchy of primary/secondary constraints. However, their original presentation focused on cases where the symplectic matrix becomes nondegenerate after a single extension, effectively restricting the method to systems with only second-class constraints.

	The full iterative power of the Faddeev--Jackiw framework was later unlocked by \textbf{Barcelos-Neto and Wotzasek}~\cite{BarcelosNetoWotzasek1992,BarcelosNetoWotzasek1993}, who introduced a systematic algorithm to handle arbitrary singular systems. Their key insight was to treat \textit{all} constraints—whether initially apparent or generated by consistency—as dynamical variables, introducing Lagrange multipliers and extending the phase space iteratively until either the extended symplectic matrix becomes invertible or a gauge symmetry is revealed~\cite{GitmanTyutin1990,Rothe1997}.

	From a computational perspective, each iteration of this algorithm corresponds to a structured matrix augmentation: the current symplectic matrix is bordered with rows and columns derived from the gradients of new constraints. This operation is formally analogous to the Matrix Bordering Technique (MBT), a well-established method in numerical linear algebra for analyzing rank deficiency, sensitivity, and regularization of linear systems~\cite{GolubVanLoan2013,StewartSun1990}. MBT is widely used in bifurcation theory~\cite{Kuznetsov2004,Seydel2010}, control~\cite{ZhouDoyle1998}, and inverse problems~\cite{Vogel2002}, precisely the domains where degeneracy signals qualitative change. Yet, while general MBT permits arbitrary blocks, the Faddeev–Jackiw iteration is geometrically constrained.
	Antisymmetry, vanishing diagonal blocks, and transposed off-diagonal coupling are not optional choices, but are enforced by the underlying symplectic structure~\cite{ArnoldBook,Marsden1999,AbrahamMarsden1978}. This restriction has profound physical consequences: it ensures that, once the invertible symplectic core is eliminated, the residual pairing of the bordering block with the null space maps canonically to the Poisson bracket algebra of constraints~\cite{Dirac1964,Henneaux1992}.

	The main contribution of this work is to establish a precise algebraic equivalence between the Faddeev-Jackiw symplectic reduction and a geometrically constrained instance of the MBT. Within this formulation, the iterative phase-space extension introduced by Barcelos-Neto and Wotzasek is shown to correspond exactly to a structured bordering procedure. Because the anchor block $f^{(0)}$ is singular by hypothesis, the classical Schur determinantal identity---which presupposes an invertible anchor---cannot be invoked; the regularity of the bordered matrix is instead controlled by how the bordering block couples to the null space of $f^{(0)}$. We prove an exact determinantal factorization (Lemma~\ref{lem:rank}, following the bordered-matrix framework of Galántai~\cite{Galantai2001}) showing that $\det(f^{(m)})$ is, up to a strictly positive factor fixed by the nonsingular symplectic core, the square of the determinant of a \emph{reduced constraint matrix} $\Gamma = N^{\top}B$ built from the pairing of the constraint gradients with a basis $N$ of $\ker(f^{(0)})$. The technical centerpiece is an \emph{exact coefficient-by-coefficient identity} between this reduced matrix and the Hessian of the symplectic potential restricted to $\ker(f^{(0)})$, which is precisely the constraint matrix whose nondegeneracy defines a second-class system. This identity reduces the regularization and termination criterion for the algorithm to the algebraic nondegeneracy of the constraint bracket.

	Recent work by Gorard~\cite{Gorard2024} reframes such geometric structures as computational primitives: the symplectic form is not merely a static object, but a dynamic rule governing state evolution in a symbolic rewriting system~\cite{Wolfram2002,WolframPhysics2020}. In this view, constraint handling becomes an instance of algorithmic geometry, where physical law emerges from the consistent application of transformation rules on data structures. Our implementation embodies this philosophy: the Faddeev-Jackiw reduction is realized as a deterministic, rule-based engine operating on immutable symbolic states.

	Previous symbolic implementations of the Faddeev–Jackiw algorithm~\cite{ArellanoPaulinCabrera2024,GomezPandeyThibes2024} were primarily aimed at concrete applications. The present work departs from that perspective and establishes a general theoretical result: we show that the FJ iteration is algebraically equivalent to a geometrically constrained form of the Matrix Bordering Technique.

	The implementation is validated through its application to representative models of singular dynamics~\cite{HojmanUrrutia1981,MukherjeeSaha1988}. These systems, characterized by nonstandard kinetic structures and nontrivial consistency conditions, provide a demanding benchmark for symplectic reduction, as the correct identification of the reduced manifold requires a faithful treatment of the underlying constraint algebra. Our algorithm determines the symplectic submanifold directly from the bordered matrix structure and reproduces the established reductions without additional assumptions. To further assess generality, we apply the method to the class of mechanical singular systems analyzed by Brown~\cite{Brown2023}, including coupled masses, rods, and pulleys. A direct comparison with the Dirac–Bergmann treatment~\cite{Dirac1964,Bergmann1949} shows that the matrix-bordering formulation offers a unified route to the emergent symplectic manifold $\mathcal{M}^*$, where constraint handling follows from matrix structure rather than from a hierarchical classification.

	Building on this foundation, we deliver a fully symbolic \textit{Mathematica} package that automates the Barcelos-Neto-Wotzasek workflow with rigorous computational design. The following sections present the mathematical theorem, its computational realization, validation examples, and broader implications for the automation of constrained dynamics.

	\section{Mathematical foundation: FJ as a geometrically constrained instance of MBT}\label{sec2}

	Consider a first-order Lagrangian of the form $L = a_i(\xi) \dot{\xi}^i - V(\xi)$~\cite{FaddeevJackiw1988,Henneaux1992}. The pre-symplectic two-form is $f_{ij} = \partial_i a_j - \partial_j a_i$. If $\det(f_{ij}) = 0$, the system implies a singular manifold~\cite{Gotay1978,Crnkovic1987}. The iterative extension of phase-space within the FJ algorithm constitutes a rigorous application of the MBT~\cite{GolubVanLoan2013,Galantai2001}. At each iteration $m$, the extended symplectic matrix takes the form:
	\[
	f^{(m)} = \begin{pmatrix} f^{(0)} & B \\ -B^\top & 0 \end{pmatrix},
	\]
	where $B_{j\alpha} = \partial \Omega_\alpha / \partial \xi^j$ represents the gradients of the constraints $\Omega_\alpha$~\cite{BarcelosNetoWotzasek1992}. The following theorem provides the mathematical backbone of the symbolic engine developed in the present work:

	\begin{theorem}[FJ as a geometrically constrained MBT instance]\label{theo1}
		Let $L = a_i(\xi)\dot{\xi}^i - V(\xi)$ be a first-order Lagrangian with pre-symplectic form
		\[
		f^{(0)}_{ij} = \partial_i a_j - \partial_j a_i,
		\]
		assumed to be singular. At iteration $m$ of the Barcelos-Neto-Wotzasek algorithm~\cite{BarcelosNetoWotzasek1992,BarcelosNetoWotzasek1993}, the extended symplectic matrix assumes the canonical bordered form:
		\[
		f^{(m)} =
		\begin{pmatrix}
			f^{(0)} & B \\
			- B^\top & 0
		\end{pmatrix},
		\quad
		\text{with} \quad B_{j\alpha} = \frac{\partial \Omega_\alpha}{\partial \xi^j},
		\]
		where $\Omega_\alpha$ are the constraints generated by the consistency conditions, assumed independent (irreducible).
		
		Let $N \in \mathbb{R}^{n\times d}$ be a matrix whose columns form an orthonormal basis of $\ker(f^{(0)})$, where $d = \dim\ker(f^{(0)}) = n - r$, and define the \emph{reduced constraint matrix}
		\[
		\Gamma \;=\; N^{\top} B \;\in\; \mathbb{R}^{d\times k},
		\qquad \Gamma_{a\alpha} \;=\; N_a^{\,i}\,\partial_i \Omega_\alpha .
		\]
		Then $f^{(m)}$ is an antisymmetric, geometrically constrained instance of the MBT~\cite{GolubVanLoan2013,StewartSun1990,Zhang2005}, and its regularity is governed exactly by $\Gamma$: one has the determinantal factorization
		\[
		\det\bigl(f^{(m)}\bigr) \;=\; \Bigl(\textstyle\prod_{\ell} \mu_\ell^{2}\Bigr)\,\det(\Gamma)^2 ,
		\qquad \mu_\ell > 0,
		\]
		where the $\mu_\ell$ are the Pfaffian factors of the nonsingular block of $f^{(0)}$. Consequently the symplectic regularity condition holds if and only if the algorithm has generated exactly $k=d$ independent constraints and the reduced matrix is nonsingular:
		\[
		\det \bigl(f^{(m)}\bigr) \neq 0
		\quad \Longleftrightarrow \quad
		k = d \ \text{and}\ \det(\Gamma) \neq 0 .
		\]
		Moreover, when the constraints arise from the consistency condition, $\Gamma$ coincides coefficient by coefficient with the Hessian of the symplectic potential restricted to $\ker(f^{(0)})$,
		\[
		\Gamma_{\alpha\beta} \;=\; N_a^{\,i}\,\bigl(\partial_i\partial_j V\bigr)\,N_b^{\,j} \;\equiv\; \{\Omega_\alpha, \Omega_\beta\}_{\mathrm{FJ}}\, ,
		\]
		which is the constraint matrix whose nondegeneracy characterizes a second-class system~\cite{Dirac1964,Henneaux1992}. Under the standing assumption that the consistency algorithm has been iterated to exhaustion with independent generated constraints, the Faddeev-Jackiw algorithm therefore terminates in a nondegenerate symplectic form precisely when the constraints form a second-class system~\cite{Dirac1964,GitmanTyutin1990}. A residual kernel ($\det\Gamma = 0$) signals first-class constraints and is treated as the gauge branch discussed in Section~\ref{sec5}.
	\end{theorem}

	\begin{proof}
		The proof proceeds in four steps: (1) identification of the constraint set and bordering block, (2) reduction of the regularity problem to the null space of $f^{(0)}$ via the reduced constraint matrix $\Gamma = N^{\top}B$, (3) an exact determinantal factorization $\det(f^{(m)}) = (\prod_\ell\mu_\ell^2)\det(\Gamma)^2$ through a congruence lemma, and (4) an exact coefficient-by-coefficient identification of the reduced matrix with the Hessian of the symplectic potential restricted to $\ker(f^{(0)})$, i.e.\ the Faddeev--Jackiw bracket matrix of the constraints. Full technical details are provided in Appendix A.
		
		\smallskip
		\textbf{Step 1: Constraint generation and bordered structure.}
		Following the rank-reduction framework for bordered matrices developed by Galántai~\cite{Galantai2001}, and building on the broader matrix-bordering perspective advocated by Kuznetsov~\cite{Kuznetsov2004} in the numerical-continuation context, a rank-deficient matrix can be regularized by augmenting it with vectors that span its null space. In the Faddeev-Jackiw formalism, the consistency condition
		\[
		v^i \partial_i V = 0,
		\qquad \forall v \in \ker(f^{(0)}),
		\]
		generates the constraint set $\Omega_\alpha(\xi)$~\cite{FaddeevJackiw1988,BarcelosNetoWotzasek1992}. The gradients of these constraints, $B_{j\alpha} = \partial_j \Omega_\alpha$, act precisely as the bordering vectors introduced at each iteration of the algorithm, yielding the bordered structure
		\[
		f^{(m)} =
		\begin{pmatrix}
			f^{(0)} & B \\
			- B^\top & 0
		\end{pmatrix}.
		\]
		This is a geometrically constrained instance of the Matrix Bordering Technique: antisymmetry and the vanishing multiplier--multiplier block are not modeling choices but algebraic consequences of the underlying symplectic structure~\cite{ArnoldBook,Marsden1999,AbrahamMarsden1978}.
		
		\smallskip
		\textbf{Step 2: Reduction to the null space.}
		The invertibility of $f^{(m)}$ cannot be analyzed through the classical Schur determinantal identity
		$\det\!\bigl(\begin{smallmatrix} A & B \\ C^\top & D \end{smallmatrix}\bigr) = \det(A)\,\det(D - C^\top A^{-1} B)$,
		since the latter requires the anchor block $A=f^{(0)}$ to be invertible, a condition violated by hypothesis. We instead diagonalize the nonsingular part of $f^{(0)}$: there is an orthogonal $Q$ with $Q^{\top}f^{(0)}Q = \mathrm{diag}(J,0)$, $J\in\mathbb{R}^{r\times r}$ invertible. Eliminating the invertible core $J$ by a congruence (Lemma~\ref{lem:rank} in Appendix~A) leaves the regularity of $f^{(m)}$ entirely controlled by the block that pairs the Lagrange multipliers with $\ker(f^{(0)})$. Concretely, writing the columns of $N\in\mathbb{R}^{n\times d}$ for an orthonormal basis of $\ker(f^{(0)})$, this residual pairing is the \emph{reduced constraint matrix}
		\begin{equation}\label{eq:gammadef}
			\Gamma \;\equiv\; N^{\top} B \;\in\; \mathbb{R}^{d\times k},
			\qquad \Gamma_{a\alpha} = N_a^{\,i}\,\partial_i\Omega_\alpha .
		\end{equation}
		It records how the constraint gradients couple to the null directions of $f^{(0)}$---the only sector that survives the elimination of the nonsingular symplectic core, and hence the sector that governs regularity.
		
		\smallskip
		\textbf{Step 3: Determinantal factorization.}
		Let $r = \mathrm{rank}(f^{(0)}) = n - d$. The congruence of Step~2 yields the exact factorization (Lemma~\ref{lem:rank})
		\begin{equation}\label{eq:detfactor}
			\det\bigl(f^{(m)}\bigr) \;=\; \Bigl(\textstyle\prod_{\ell=1}^{r/2}\mu_\ell^{2}\Bigr)\,\det(\Gamma)^2 ,
		\end{equation}
		where $\mu_\ell>0$ are the Pfaffian factors of $J$. Since the prefactor is strictly positive, $f^{(m)}$ is nonsingular if and only if (i) $k = d = n-r$, and (ii) $\det(\Gamma)\neq 0$, i.e.\ $\mathrm{rank}(\Gamma)=d$. Condition (ii) is the precise meaning of the informal phrase \emph{``$B$ is structurally compatible with $\ker(f^{(0)})$''}: the constraint gradients must pair non-degenerately with the null space of the pre-symplectic form. Hence
		\begin{equation}\label{eq:schurcriterion}
			\det(f^{(m)}) \neq 0 \;\iff\; \det(\Gamma) \neq 0 .
		\end{equation}
		
		\smallskip
		\textbf{Step 4: Exact kernel--Poisson identity.}
		When the constraints are generated by the consistency condition, so that $\Omega_\alpha = N_\alpha^{\,i}\partial_i V$ up to independent recombination, a direct computation from Definition~\eqref{eq:gammadef} gives
		\begin{equation}\label{eq:schurpoisson}
			\Gamma_{\alpha\beta} \;=\; \sum_{i,j=1}^{n} N_\alpha^{\,i}\,\bigl(\partial_i\partial_j V\bigr)\,N_\beta^{\,j} \;=\; \{\Omega_\alpha,\Omega_\beta\}_{\mathrm{FJ}},
		\end{equation}
		where the Faddeev--Jackiw bracket of the constraints is identified with the Hessian of the symplectic potential restricted to $\ker(f^{(0)})$. The identification~\eqref{eq:schurpoisson} holds \emph{coefficient by coefficient}: it is an exact matrix equality, not an asymptotic equivalence, an isomorphism of algebraic structures, or a relation modulo terms vanishing on the constraint surface. In particular, the Dirac second-class constraint matrix coincides literally with the reduced matrix: $\mathcal{C}_{\alpha\beta} = \Gamma_{\alpha\beta}$.
		
		\smallskip
		\textbf{Conclusion.} Combining~\eqref{eq:schurcriterion} and~\eqref{eq:schurpoisson} establishes the biconditional
		\[
		\det(f^{(m)}) \neq 0 \;\iff\; \det(\Gamma) \neq 0 \;\iff\; \det\{\Omega_\alpha,\Omega_\beta\}_{\mathrm{FJ}} \neq 0.
		\]
		The forward implication is immediate from~\eqref{eq:detfactor}. For the converse, $\det\{\Omega_\alpha,\Omega_\beta\}_{\mathrm{FJ}}\neq 0$ forces $\det(\Gamma)\neq 0$ by~\eqref{eq:schurpoisson}, hence $\mathrm{rank}(\Gamma)=d$ and, by~\eqref{eq:detfactor}, $f^{(m)}$ has full rank $n+k=n+d$. Under the standing assumption that the consistency algorithm has been iterated to exhaustion with independent generated constraints, the termination of the Faddeev-Jackiw iteration in a nondegenerate symplectic form is therefore controlled entirely by the symplectic closure of the constraint algebra~\cite{Gotay1978,Crnkovic1987,Dirac1964,GitmanTyutin1990}.
	\end{proof}
	
	From a linear-algebraic standpoint, the regularization achieved at termination corresponds to the attainment of full rank in a previously singular matrix via a structured extension~\cite{Galantai2001,HornJohnson2013}. Such mechanisms are closely related to classical results on rank reduction and bordered inversion, where the invertibility of an augmented matrix is controlled by its Schur complement~\cite{Zhang2005,Cottle2009}. In particular, bordered inversion can be understood as the dual operation to rank reduction in a purely algebraic setting~\cite{Galantai2001}. The present construction may therefore be viewed as a geometrically constrained realization of this general principle, in which antisymmetry and symplectic structure~\cite{ArnoldBook,Marsden1999} restrict the admissible bordering blocks arising from singular Lagrangian systems.
	
	\begin{remark}
		A subtle but crucial consequence of formulating the Faddeev-Jackiw iteration as a geometrically constrained instance of the MBT is that the regularization is structural rather than numerical~\cite{Kuznetsov2004,Seydel2010}. The algorithm does not ``cancel'' singularities by absorbing global factors or parameters; instead, it propagates them symbolically through each bordering step. As a result, parametric dependencies are preserved at the level of the constraint algebra itself~\cite{Henneaux1992,Rothe1997}. This guarantees that {potential} bifurcation conditions, degeneracy loci, and stability thresholds remain visible in the final symplectic manifold $\mathcal{M}^*$, rather than being obscured by premature simplifications~\cite{Kuznetsov2004,Guckenheimer1983}. From a computational perspective, the reduction process is therefore best understood as a rule-governed transformation of symbolic states~\cite{Wolfram2002,Gorard2024}, not as a sequence of ad hoc algebraic manipulations.
	\end{remark}

	\section{Computational architecture and symbolic implementation}\label{sec3}

	The symbolic realization of the Barcelos-Neto-Wotzasek algorithm is framed as a deterministic rewriting system~\cite{Wolfram2002,WolframPhysics2020}. In this paradigm, the physical system is represented by a state vector $\mathcal{S} = \{ \xi, L, f, V \}$, where $f$ is the pre-symplectic form. Each iteration of the method constitutes a computational event that transforms $\mathcal{S}$ into a more constrained state $\mathcal{S}'$ until the symplectic manifold $\mathcal{M}^*$ emerges~\cite{Gorard2024,FaddeevJackiw1988}.

	\subsection{Structural graph of the reduction process}

	To make explicit how the abstract matrix bordering theorem translates into an executable symbolic workflow, we formalize the reduction procedure as a \textbf{directed dependency graph}. Drawing inspiration from the algorithmic paradigms explored in the Wolfram Physics Project \cite{Gorard2024,WolframPhysics2020}---where such graphs are termed \emph{causal} in the rewriting-system sense---we define a framework where the nodes represent distinct structural states of the phase space (such as singular manifolds or extended matrices), while the edges denote the deterministic computational rules (specifically matrix bordering and null-space extraction) that drive the system's evolution~\cite{Wolfram2002}. This visualization serves as a structural map, elucidating how the algorithm navigates from an initial singular Lagrangian to the emergent symplectic manifold $\mathcal{M}^*$.

	This graph-theoretic representation exposes a fundamental duality: the computational progression through bordering iterations is isomorphic to a structural descent along constraint manifolds~\cite{Gotay1978,ArnoldBook,Marsden1999}. Each state transition encodes a reduction in the kernel dimension of the symplectic form, thereby shrinking the accessible phase space until either a regular symplectic structure emerges or an irreducible gauge orbit is encountered~\cite{Henneaux1992,MarsdenWeinstein1974}. Critically, the graph is acyclic in the absence of gauge redundancies, the dashed feedback loop closes only when the determinant criterion fails, signaling that the current extension has induced a consistent Lagrangian redefinition without eliminating the singularity. In this sense the dependency graph is not merely a visualization aid but a structural descriptor of the algorithm~\cite{Gorard2024,Wolfram2002}: it captures the minimal sequence of bordering operations required to either regularize the system or expose its gauge content. This viewpoint parallels the rewriting-system perspective of the Wolfram Physics Project~\cite{WolframPhysics2020}, in which the dependency (``causal'') graph of the rules encodes the evolution. Here, the emergence of $\mathcal{M}^*$ (when attainable) follows from the iterative enforcement of symplectic closure on an initially degenerate state space.

	\begin{figure}[htbp]
		\centering
		\begin{tikzpicture}[node distance=2.2cm, auto, thick]
			
			% Nodos
			\node[state] (L0) {$\mathcal{L}^{(0)}$};
			
			% CORRECCIÓN 1: Aumentamos la distancia específica aquí a 3.5cm
			\node[state, right=3.5cm of L0] (Fk) {$f_{ij}^{(k)}$};
			
			% CORRECCIÓN 2: Nodo Kernel visualmente más ligero (auxiliar)
			\node[state,
			below=of Fk,
			draw opacity=0.6, % Borde semitransparente
			text opacity=0.85, % Texto ligeramente suave
			line width=0.8pt] (Null) {$\mathcal{K}^{(k)}$};
			
			\node[state, right=of Null] (Omega) {$\Omega^{(k)}$};
			
			\node[state, above=of Omega] (Ext) {$\xi^{(k+1)}$};
			
			\node[state, final, right=of Ext, xshift=1cm] (Final) {$\mathcal{M}^*$};
			
			% Transiciones
			
			% Etiqueta partida en dos líneas para que no choque
			\draw[transition]
			(L0) -- node[edge label, align=center] {Initial\\symplectic data} (Fk);
			
			\draw[transition]
			(Fk) -- node[left, font=\footnotesize] {$\det(f)=0$} (Null);
			
			\draw[transition]
			(Null) -- node[below, font=\footnotesize] {$v \in \ker(f)$} (Omega);
			
			\draw[transition]
			(Omega) -- node[right, font=\footnotesize] {$\dot{\Omega} = 0$} (Ext);
			
			% CORRECCIÓN 3: Realimentación curva y punteada (densely dashed)
			% Representa redefinición inducida, no flujo simple
			\draw[transition, densely dashed]
			(Ext) to[out=160,in=20]
			node[above, font=\footnotesize, yshift=2pt]
			{$\mathcal{L}^{(k)} + \dot{\lambda}\,\Omega$} (Fk);
			
			\draw[transition]
			(Ext) -- node[edge label] {$\det(f) \neq 0$} (Final);
			
			% Caption interno mejorado
			\node[text width=10cm, align=left, font=\scriptsize\itshape,
			below=0.8cm of Omega, xshift=-0.5cm]
			{Evolution of the symbolic state $\mathcal{S}$ across successive phase-space extensions.
				Nodes represent structural states, while edges encode deterministic operations
				(matrix bordering and null-space extraction). The dashed loop denotes an induced
				redefinition of the symplectic 2-form when regularization fails.};
		\end{tikzpicture}
		\caption{Summarizes the Faddeev–Jackiw reduction as a dependency graph, emphasizing its interpretation as a geometrically constrained matrix bordering process rather than a procedural flowchart.}
		\label{fig:causal_graph}
	\end{figure}

	\section{Results and technical discussion}\label{sec4}

	The computational core is validated through a tripartite analysis, transitioning from abstract models to complex mechanical configurations. We explicitly present the \textit{Inverse Extended Matrix} for each regularizable case, demonstrating the engine's ability to recover the fundamental brackets~\cite{Dirac1964,Henneaux1992}.

	\subsection{Benchmark I: Singular system with noncanonical kinetics}
	%The Hojman-Urrutia Model}

We implemented a singular Lagrangian system reported in Ref.~\cite{HojmanUrrutia1981}, characterized by a non-standard kinetic term~\cite{MukherjeeSaha1988,Montani1993}. Our symbolic engine processed the state in a single bordering event. The algorithm identifies the hidden symplectic structure and returns the exact $4 \times 4$ inverse matrix, verifying the system is second-class~\cite{Dirac1964}:

\begin{equation}
	(f^{(1)})^{-1} = 
	\begin{pmatrix}
		0 & 0 & 1 & 0 \\
		0 & 0 & 0 & -1 \\
		-1 & 0 & 0 & 1 \\
		0 & 1 & -1 & 0
	\end{pmatrix}
\end{equation}
This result confirms that the initial singularity is fully resolved by the geometric constraints inherent in the model~\cite{HojmanUrrutia1981}.

\subsection{Benchmark II: Coupled masses and rigid rods}

The four-mass system (Fig. \ref{fig:brown_masses}) imposes holonomic constraints via rigid rods~\cite{Goldstein2002,Arnold1978}. In the Dirac-Bergmann framework \cite{Brown2023,Dirac1964,Bergmann1949}, this requires a multi-stage hierarchy. In our approach, the rod lengths act as bordering vectors. 

\begin{figure}[ht]
	\centering
	\begin{tikzpicture}[
		scale=2,
		% Definición de la perspectiva (aproximación isométrica)
		x={(-0.5cm,-0.3cm)},
		y={(1cm,0cm)},
		z={(0cm,1cm)},
		line cap=round,
		line join=round
		]
		
		%%%%%%%%%%%%%%%%%%%%%%%%%%%%%%%%%%%%%%%%%%%%%%%%%%%%%%%%%
		% 1) Coordenadas de las láminas (placas) originales
		%%%%%%%%%%%%%%%%%%%%%%%%%%%%%%%%%%%%%%%%%%%%%%%%%%%%%%%%%
		% Placa inferior (en z=0)
		\coordinate (B1) at (0,0,0);
		\coordinate (B2) at (3,0,0);
		\coordinate (B3) at (3,2,0);
		\coordinate (B4) at (0,2,0);
		
		% Placa superior (en z=2)
		\coordinate (T1) at (0,0,2);
		\coordinate (T2) at (3,0,2);
		\coordinate (T3) at (3,2,2);
		\coordinate (T4) at (0,2,2);
		
		%%%%%%%%%%%%%%%%%%%%%%%%%%%%%%%%%%%%%%%%%%%%%%%%%%%%%%%%%
		% 2) Dibujamos las láminas (marrón semitransparente)
		%%%%%%%%%%%%%%%%%%%%%%%%%%%%%%%%%%%%%%%%%%%%%%%%%%%%%%%%%
		\fill[gray!50,draw=black] (B1) -- (B2) -- (B3) -- (B4) -- cycle;
		\fill[gray!50,draw=black] (T1) -- (T2) -- (T3) -- (T4) -- cycle;
		
		%%%%%%%%%%%%%%%%%%%%%%%%%%%%%%%%%%%%%%%%%%%%%%%%%%%%%%%%%
		% 3) Desplazar los postes hacia el interior de las láminas
		%    (los postes se moverán a lo largo de la recta que une
		%     cada vértice con el centro de la placa)
		%%%%%%%%%%%%%%%%%%%%%%%%%%%%%%%%%%%%%%%%%%%%%%%%%%%%%%%%%
		% Centro de la placa inferior y superior
		\coordinate (Bcenter) at (1.5,1,0);
		\coordinate (Tcenter) at (1.5,1,2);
		
		% Factor de desplazamiento (0: en el vértice; 1: en el centro)
		\def\alpha{0.2}
		
		% Nuevas coordenadas para los postes (vértices desplazados)
		\coordinate (B1a) at ($(B1)! \alpha ! (Bcenter)$);
		\coordinate (B2a) at ($(B2)! \alpha ! (Bcenter)$);
		\coordinate (B3a) at ($(B3)! \alpha ! (Bcenter)$);
		\coordinate (B4a) at ($(B4)! \alpha ! (Bcenter)$);
		
		\coordinate (T1a) at ($(T1)! \alpha ! (Tcenter)$);
		\coordinate (T2a) at ($(T2)! \alpha ! (Tcenter)$);
		\coordinate (T3a) at ($(T3)! \alpha ! (Tcenter)$);
		\coordinate (T4a) at ($(T4)! \alpha ! (Tcenter)$);
		
		% Dibujamos los postes (líneas verticales en gris)
		%\draw[line width=2pt, line cap=round, gray] (B1a) -- (T1a);
		\draw[line width=2pt, line cap=round, gray!40,
		preaction={draw, line width=3pt, black}]
		(B1a) -- (T1a);
		%\draw[line width=2pt, line cap=round, gray] (B2a) -- (T2a);
		\draw[line width=2pt, line cap=round, gray!40,
		preaction={draw, line width=3pt, black}]
		(B2a) -- (T2a);
		%\draw[line width=2pt, line cap=round, gray] (B3a) -- (T3a);
		\draw[line width=2pt, line cap=round, gray!40,
		preaction={draw, line width=3pt, black}]
		(B3a) -- (T3a);
		%\draw[line width=2pt, line cap=round, gray] (B4a) -- (T4a);
		\draw[line width=2pt, line cap=round, gray!40,
		preaction={draw, line width=3pt, black}]
		(B4a) -- (T4a);
		
		%%%%%%%%%%%%%%%%%%%%%%%%%%%%%%%%%%%%%%%%%%%%%%%%%%%%%%%%%
		% 4) Definir alturas distintas para los puntos y1..y4
		%    (Estos son los puntos en los postes donde terminan los resortes)
		%%%%%%%%%%%%%%%%%%%%%%%%%%%%%%%%%%%%%%%%%%%%%%%%%%%%%%%%%
		\pgfmathsetmacro{\fYone}{0.70} % mayor fracción: más abajo
		\pgfmathsetmacro{\fYtwo}{0.60} % menor: más arriba
		\pgfmathsetmacro{\fYthree}{0.65}
		\pgfmathsetmacro{\fYfour}{0.55}
		
		\coordinate (Y1) at ($(T1a)!\fYone!(B1a)$);
		\coordinate (Y2) at ($(T2a)!\fYtwo!(B2a)$);
		\coordinate (Y3) at ($(T3a)!\fYthree!(B3a)$);
		\coordinate (Y4) at ($(T4a)!\fYfour!(B4a)$);
		
		%%%%%%%%%%%%%%%%%%%%%%%%%%%%%%%%%%%%%%%%%%%%%%%%%%%%%%%%%
		% 5) Dibujar los resortes (en azul) desde la parte superior
		%    de cada poste hasta los puntos y_i
		%%%%%%%%%%%%%%%%%%%%%%%%%%%%%%%%%%%%%%%%%%%%%%%%%%%%%%%%%
		\draw[blue, decorate, decoration={coil, segment length=5pt, amplitude=4pt}] (T1a) -- (Y1);
		\draw[blue, decorate, decoration={coil, segment length=5pt, amplitude=4pt}] (T2a) -- (Y2);
		\draw[blue, decorate, decoration={coil, segment length=5pt, amplitude=4pt}] (T3a) -- (Y3);
		\draw[blue, decorate, decoration={coil, segment length=5pt, amplitude=4pt}] (T4a) -- (Y4);
		
		%%%%%%%%%%%%%%%%%%%%%%%%%%%%%%%%%%%%%%%%%%%%%%%%%%%%%%%%%
		% 6) Colocar puntos negros en y_i con sus etiquetas
		%%%%%%%%%%%%%%%%%%%%%%%%%%%%%%%%%%%%%%%%%%%%%%%%%%%%%%%%%
		\fill (Y1) circle (0.8pt)
		node[above=3pt,left] {$y_1$};
		
		\fill (Y2) circle (0.8pt)
		node[left] {$y_2$};
		
		% y3: desplazar etiqueta hacia "afuera" y arriba
		\fill (Y3) circle (0.8pt)
		node at ($(Y3)+(0.58,0.14,0.05)$) {$y_3$};
		
		% y4: desplazar etiqueta hacia el frente-izquierda y arriba
		\fill (Y4) circle (0.8pt)
		node at ($(Y4)+(-0.12,0.12,0)$) {$y_4$};
		
		%%%%%%%%%%%%%%%%%%%%%%%%%%%%%%%%%%%%%%%%%%%%%%%%%%%%%%%%%
		% 7) Conectar los puntos y_i con varillas verdes (formando el perímetro)
		%%%%%%%%%%%%%%%%%%%%%%%%%%%%%%%%%%%%%%%%%%%%%%%%%%%%%%%%%
		
		%Original:
		%\draw[very thick, teal] (Y1) -- (Y2) -- (Y3) -- (Y4) -- cycle;
		%New
		\draw[thick, double=teal!50] (Y1) -- (Y2) -- (Y3) -- (Y4) -- cycle;
		
		%%%%%%%%%%%%%%%%%%%%%%%%%%%%%%%%%%%%%%%%%%%%%%%%%%%%%%%%%
		% 8) Colocar bolas rojas (masas) en la mitad de cada varilla verde
		%%%%%%%%%%%%%%%%%%%%%%%%%%%%%%%%%%%%%%%%%%%%%%%%%%%%%%%%%
		\foreach \A/\B in {Y1/Y2, Y2/Y3, Y3/Y4, Y4/Y1} {
			\coordinate (Mid) at ($(\A)!0.5!(\B)$);
			\shade[ball color=red] (Mid) circle (3pt);
		}
		
	\end{tikzpicture}
	\caption{Geometric configuration of the four-mass singular system. The rigid rod constraints trigger the matrix bordering process.}
	\label{fig:brown_masses}
\end{figure}

The symbolic engine generates an $8 \times 8$ extended symplectic matrix~\cite{BarcelosNetoWotzasek1992,Henneaux1992}. The engine returns the eight augmented symplectic variables ordered as
\[
\xi^{(1)} = \bigl(\,y_1,\ y_2,\ y_3,\ y_4,\ p_1,\ p_2,\ p_3,\ \lambda\,\bigr),
\]
i.e.\ the four positions $y_i$, three momenta $p_i$, and the single Lagrange multiplier $\lambda$ introduced by the rigid-rod constraint, which together account for the eight rows and columns of the displayed inverse. The computed inverse then reveals the coupled dynamics between these variables. Specifically, the exact form obtained is:

\begin{equation}
	\resizebox{0.75\textwidth}{!}{%
		$
		(f^{(1)})^{-1} = 
		\begin{pmatrix}
			0 & 0 & 0 & 0 & 3/4 & 1/4 & -1/4 & -1/4k \\
			0 & 0 & 0 & 0 & 1/4 & 3/4 & 1/4 & 1/4k \\
			0 & 0 & 0 & 0 & -1/4 & 1/4 & 3/4 & -1/4k \\
			0 & 0 & 0 & 0 & 1/4 & -1/4 & 1/4 & 1/4k \\
			-3/4 & -1/4 & 1/4 & -1/4 & 0 & 0 & 0 & 0 \\
			-1/4 & -3/4 & -1/4 & 1/4 & 0 & 0 & 0 & 0 \\
			1/4 & -1/4 & -3/4 & -1/4 & 0 & 0 & 0 & 0 \\
			1/4k & -1/4k & 1/4k & -1/4k & 0 & 0 & 0 & 0 \\
		\end{pmatrix}
		$
	}
\end{equation}

\subsection{Benchmark III: Masses on a ring with spring coupling}

The final case involves three masses $(\theta_1, \theta_2, \theta_3)$ constrained to a ring and coupled via linear springs~\cite{Goldstein2002,Landau1976}. Within the Dirac--Bergmann framework~\cite{Brown2023,Dirac1964}, this system exhibits cyclic symmetry and a pronounced parametric dependence on the spring constant $k$ and the ring radius $R$.

\begin{figure}[ht]
	\centering
	\begin{tikzpicture}[line cap=round, line join=round, >=latex, scale=2]
		\def\R{1.2}
		\draw[thick, teal, line width=3pt] (0,0) circle (\R);
		\fill[black] (0,0) circle (1.4pt);
		\draw[->, thick] (0, \R + 0.05) -- (0, \R + 0.5) node[above] {$y$};
		\draw[->, thick] (\R + 0.05, 0) -- (\R + 0.5, 0) node[right] {$x$};
		
		\def\thetaOne{30} \def\thetaTwo{150} \def\thetaThree{270}
		\coordinate (M1) at (\thetaOne:\R); \coordinate (M2) at (\thetaTwo:\R); \coordinate (M3) at (\thetaThree:\R);
		
		\draw[blue, decorate, decoration={coil, aspect=0.3, segment length=2mm, amplitude=3mm, pre length=3pt, post length=3pt}] (0,0) -- (M1);
		\draw[blue, decorate, decoration={coil, aspect=0.3, segment length=2mm, amplitude=3mm, pre length=3pt, post length=3pt}] (0,0) -- (M2);
		\draw[blue, decorate, decoration={coil, aspect=0.3, segment length=2mm, amplitude=3mm, pre length=3pt, post length=3pt}] (0,0) -- (M3);
		
		\foreach \point/\color in {M1/red, M2/blue, M3/green}{
			\shade[ball color=\color, draw=black] (\point) circle (4pt);
		}
		
		\node[above right=2pt] at (M1)  {$\theta_1$};
		\node[above left=3pt]  at (M2)  {$\theta_2$};
		\node[below left=1pt]  at (M3)  {$\theta_3$};
		\draw[->] (1.4,0) arc (0:30:1.1);
	\end{tikzpicture}
	\caption{Three masses on a ring connected by springs. The system requires two iterations of the FJ algorithm to reach $\mathcal{M}^*$.}
	\label{fig:ring_system}
\end{figure}

The symbolic engine required two iterations (\textit{IterationCount: 2}) to fully regularize the matrix. We present the final inverse matrix in a factored form to highlight its structural symmetries~\cite{ArnoldBook,Marsden1999}, using the shorthand $S_i \equiv \sin(\theta_i)$ and $C_i \equiv \cos(\theta_i)$:

\begin{equation}
	\label{eq:inverse_ring}
	\resizebox{0.75\textwidth}{!}{%
		$
		(f^{(2)})^{-1} = \frac{1}{3}
		\left(
		\begin{array}{cccccccccc}
			0 & 0 & 0 & 0 & 0 & 3 & 0 & 0 & 0 & 0 \\
			0 & 0 & 0 & 0 & 0 & 0 & 3 & 0 & 0 & 0 \\
			0 & 0 & 0 & 0 & 0 & 0 & 0 & 3 & 0 & 0 \\
			0 & 0 & 0 & 0 & 0 & -R S_1 & -R S_2 & -R S_3 & 0 & -\frac{1}{k} \\
			0 & 0 & 0 & 0 & 0 & R C_1 & R C_2 & R C_3 & -\frac{1}{k} & 0 \\
			-3 & 0 & 0 & R S_1 & -R C_1 & 0 & 0 & 0 & 0 & 0 \\
			0 & -3 & 0 & R S_2 & -R C_2 & 0 & 0 & 0 & 0 & 0 \\
			0 & 0 & -3 & R S_3 & -R C_3 & 0 & 0 & 0 & 0 & 0 \\
			0 & 0 & 0 & 0 & \frac{1}{k} & 0 & 0 & 0 & 0 & 0 \\
			0 & 0 & 0 & \frac{1}{k} & 0 & 0 & 0 & 0 & 0 & 0 \\
		\end{array}
		\right)
		$
	}
\end{equation}

\subsection{On the preservation of parametric factors}

A crucial aspect of our implementation is the deliberate preservation of global factors and physical parameters (such as $k$) throughout the symbolic reduction~\cite{Kuznetsov2004,Seydel2010}.
% Cambio desde donde dice This ensure...
As discussed in the context of bifurcation theory \cite{Kuznetsov2004,Guckenheimer1983}, the ability to detect branching points depends on the analytic properties of the defining system. In our framework, the symplectic matrix is regularized via structure-preserving operations~\cite{ArnoldBook,Marsden1999,Hairer2006} that retain the complete constraint architecture and its parametric dependencies. This ensures that the emergent brackets faithfully reflect the underlying parameter space~\cite{Henneaux1992,Rothe1997}, allowing the researcher to identify bifurcation conditions when critical parameters are varied, as well as to assess structural stability in regimes where no such bifurcations occur~\cite{Kuznetsov2004,Wiggins2003}. By treating parameters as symbolic invariants, the final manifold $\mathcal{M}^*$ serves as a robust foundation for stability analysis.

\section{Computational design and symbolic interface}\label{sec5}

The bridge between the geometric theory of constrained systems and its algorithmic realization lies in the design of the symbolic interface~\cite{Wolfram2002,Gorard2024}. Rather than treating the Lagrangian as a mere mathematical expression to be parsed procedurally, our implementation adopts a {declarative data-driven paradigm}. We define the physical system not as a function of time, but as a static association of structural properties mapping the tangent bundle $TQ$ to the symplectic manifold~\cite{ArnoldBook,Marsden1999}.

\subsection{The declarative system definition}

In standard computer algebra systems (CAS)~\cite{Fateman1991,Geddes1992}, Lagrangians are often defined with explicit time dependencies (e.g., $q_1(t), \dot{q}_1(t)$), enforcing the symbolic engine to carry the overhead of functional calculus throughout the reduction. We depart from this convention by treating phase space coordinates as atomic algebraic symbols. The temporal evolution is implicit, emerging only when the {exterior derivative operator} acts upon the symplectic structure~\cite{ArnoldBook,Nakahara2003}.

This design choice drastically reduces the syntactic noise and computational complexity~\cite{Wolfram2002}. The user defines the system as an \texttt{Association}, a key-value data structure that encapsulates the kinetic term, the symplectic potential, and the configuration variables. This separation allows the bordering engine to compute canonical momenta and Hessian blocks independently before assembling the pre-symplectic two-form~\cite{FaddeevJackiw1988,BarcelosNetoWotzasek1992}.

\subsection{Input syntax and implicit dynamics}

To illustrate the compactness of the framework, we present the declaration of the singular system. The absence of explicit time arguments \texttt{[t]} reflects that the engine treats the velocities as independent atomic symbols, i.e.\ as coordinates on the first-order jet bundle $J^1(Q)$ (equivalently the tangent bundle $TQ$), on which the pre-symplectic one-form is defined; no variational calculus is performed at this stage~\cite{Olver1993,Saunders1989}.

\begin{figure}[h]
	\flushleft
	% Aquí va la definición de Hojman-Urrutia
	\includegraphics[width=0.65\textwidth]{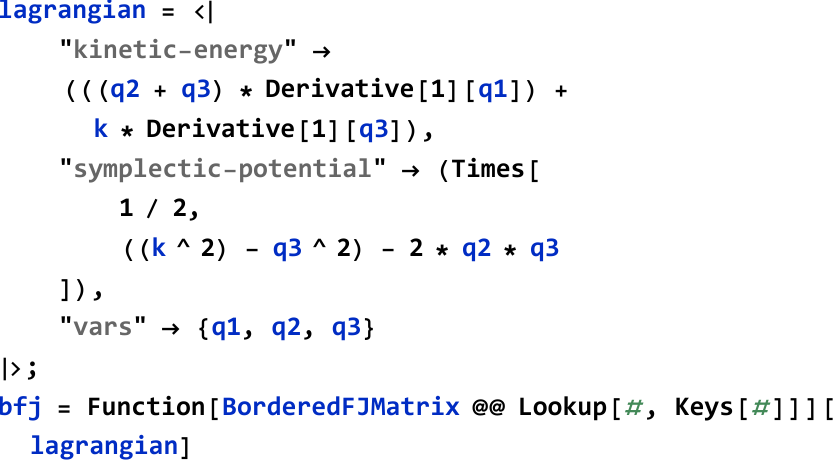} 
	\label{hojman}
\end{figure}

In this framework, the velocities are not represented as differentiated dynamical functions $\dot q_i(t)$ but as independent fiber coordinates on the jet bundle~\cite{Olver1993,Saunders1989}. As evidenced in the symbolic declaration of the above system, the kinetic sector is defined using the native Wolfram syntax \texttt{Derivative[1][q1]} and \texttt{Derivative[1][q3]}, which serve here as atomic symbols standing for $\dot{q}_1$ and $\dot{q}_3$, respectively. The \texttt{BorderedFJMatrix} function ingests this immutable state directly from the \texttt{Association}. Treating these symbols as coordinates rather than as operators, the engine constructs the initial pre-symplectic matrix $f^{(0)}$ and recursively applies the structure-preserving bordering theorem derived in Section \ref{sec2}, ensuring that the reduction process remains faithful to the underlying manifold geometry~\cite{Marsden1999,ArnoldBook}.

\subsection{Structured output and symbolic encapsulation}

Handling the symplectic geometry of high-dimensional systems generates massive algebraic expressions that can easily clutter the workspace. To mitigate this, our engine adheres to the \textit{Principle of Minimal Cognitive Load}~\cite{Sweller1988,Chandler1991}. Upon convergence, the function returns an \textbf{opaque symbolic object} formatted as a standard Wolfram Language \texttt{SummaryBox}.

When the system is successfully regularized, this interface provides an immediate, high-level view of the reduction's outcome (see Fig. \ref{fig:summarybox}). The summary box displays the principal structural data of the reduction:
\begin{itemize}
	\item The \textbf{Regularity Status} (indicating a successful transition to a symplectic manifold~\cite{ArnoldBook}).
	\item The \textbf{Extended Dimension} of the phase space ($2N \times 2N$).
	\item The \textbf{Constraint Count} and \textbf{Iteration Depth}.
\end{itemize}

\begin{figure}[h]
	\flushleft
	% Aquí iría tu captura de la caja bonita
	\includegraphics[width=0.65\textwidth]{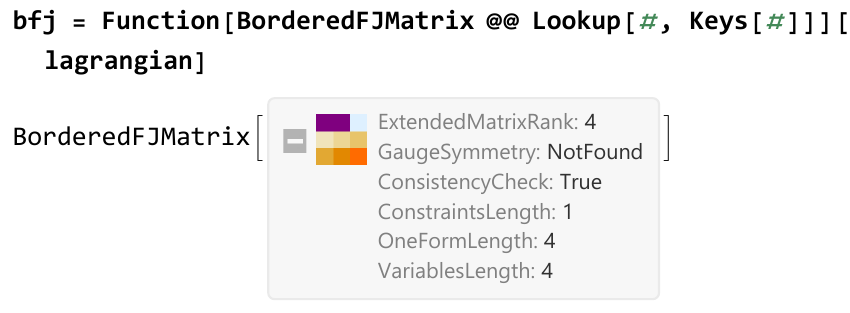} 
	\caption{The output object displayed as a SummaryBox. It encapsulates the full symplectic hierarchy while exposing only essential metadata. The ``+'' icon allows for interactive drill-down into the constraint algebra without recomputing the manifold.}
	\label{fig:summarybox}
\end{figure}

To reveal the internal architecture of the computed manifold, the user can query the object using the \texttt{"Properties"} key. This command generates a comprehensive list of the symplectic and algebraic components stored within the encapsulated state (see Fig. \ref{fig:properties_list}).

\begin{figure}[h]
	\flushleft
	\includegraphics[width=0.35\textwidth]{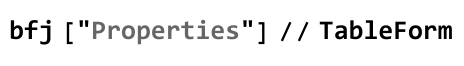}\\[0.2em]
	\includegraphics[width=0.35\textwidth]{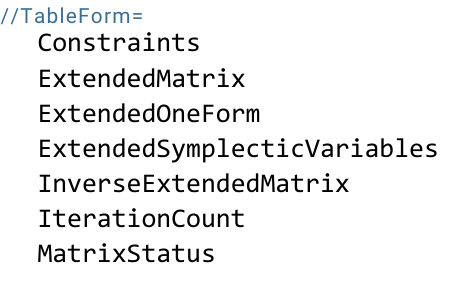}
	\caption{Interactive query of the object's metadata. 
		Top: summary of available properties returned by the \texttt{"Properties"} key. 
		Bottom: internal structure exposed through the association interface.}
	\label{fig:properties_list}
\end{figure}

Despite this concise visual representation, the object behaves as a queryable association. Underlying the summary interface, the full structural state of the final manifold $\mathcal{M}^*$ is preserved and accessible via specific string-based keys:

\begin{itemize}
	\item \texttt{"Constraints"}: A list of the geometric constraints $\Omega_\alpha$ generated sequentially via the consistency conditions~\cite{Dirac1964,BarcelosNetoWotzasek1992}.
	\item \texttt{"ExtendedMatrix"}: The final symplectic two-form $f^{(m)}$ bordering the original sector with the constraint gradients~\cite{FaddeevJackiw1988}.
	\item \texttt{"ExtendedOneForm"}: The canonical one-form $a^{(m)}_i \dd \xi^i$ associated with the regularized Lagrangian~\cite{ArnoldBook,Marsden1999}. This vector encapsulates the modified symplectic potentials, including the terms derived from the Lagrange multipliers.
	\item \texttt{"ExtendedSymplecticVariables"}: The augmented phase space vector $\xi^{(m)}$, including the necessary Lagrange multipliers required to close the algebra~\cite{Henneaux1992,GitmanTyutin1990}.
	%Cambiado 29/01/2026
	\item \texttt{"InverseExtendedMatrix"}: The crucial outcome of the regularization process, providing the generalized symplectic brackets of the theory~\cite{Dirac1964,FaddeevJackiw1988}.
	%Cambiado 29/01/2026
	\item \texttt{"IterationCount"}: An integer specifying the iteration step at which a regular symplectic matrix is obtained.
	\item \texttt{"MatrixStatus"}: A status flag (\texttt{"Regular"} or \texttt{"Singular"}) confirming the invertibility of the 2-form and the successful reach of the symplectic manifold~\cite{Gotay1978,ArnoldBook}.
\end{itemize}

This design decouples the representation of the physical system from its implementation~\cite{Wolfram2002,Gorard2024}. For instance, extracting the fundamental brackets reduces to a simple key lookup:
\begin{verbatim}
	(* Extracting the Generalized Symplectic Brackets *)
	DiracStructure = bfj["InverseExtendedMatrix"];
	
	(* Analyzing the constraint algebra *)
	PhysicalConstraints = bfj["Constraints"];
\end{verbatim}
%\vspace{-1em}
By returning a comprehensive state object, the engine facilitates seamless integration with downstream analysis tools, such as stability analyzers or tensor calculus packages~\cite{mgarcia,Nakahara2003}.

\subsection{Automated detection of gauge symmetries}

A critical feature of the matrix bordering engine is its ability to flag the presence of candidate gauge redundancies~\cite{Henneaux1992,Rothe1997}. In the Faddeev-Jackiw framework, a system that fails to reach a regular symplectic manifold $\mathcal{M}^*$ after the exhaustion of consistency conditions retains a nontrivial kernel of the final pre-symplectic form, which is a necessary signature of first-class constraints in the Dirac-Bergmann nomenclature~\cite{Dirac1964,GitmanTyutin1990}. We emphasize, however, that a residual kernel is not by itself sufficient to establish a gauge symmetry: the surviving null modes $v$ must additionally satisfy the consistency condition $v^i\partial_i V = 0$ so that no further constraint is generated (if this fails, the halt reflects an \emph{inconsistent} system rather than a gauge one), and the associated constraints must be verified to be first-class. Accordingly, our engine \emph{detects and exposes the candidate null generators} $v\in\ker(f^{(m)})$, leaving to the user the physical confirmation that they generate genuine symmetries of the action. The reduction has three mutually exclusive outcomes at a halt: a fixed Lagrange multiplier (second-class direction), a null mode satisfying $v^i\partial_i V=0$ (gauge direction), or a null mode with $v^i\partial_i V\neq 0$ (inconsistency).

Our implementation monitors the iterative process and halts when it detects that no further independent geometric information can be extracted from the Lagrangian. The engine reports these structural signatures through two specific diagnostic messages:

\begin{itemize}
	\item \textbf{``Null constraint detected: all new constraints are identically zero.''} This occurs when the consistency condition $\dot{\Omega} = 0$ is satisfied identically on the current manifold, implying that the existing constraints are already stable under the Hamiltonian flow~\cite{ArnoldBook,Marsden1999}.
	\item \textbf{``Dependent constraints detected: new constraints are linearly dependent on existing ones.''} This indicates that while new constraints are generated, they do not restrict the phase space further, revealing a redundancy in the constraint algebra~\cite{Henneaux1992,GitmanTyutin1990}.
\end{itemize}

To illustrate this diagnostic capability, we analyzed a variant of the four-mass system introduced by Brown~\cite{Brown2023} (see Fig.~\ref{fig:gauge_system}), in which certain elastic couplings are removed, thereby introducing a global translational symmetry~\cite{Goldstein2002,Marsden1999}.

\begin{figure}[H]
	\centering
	\begin{tikzpicture}[
		scale=2,
		% Perspectiva isométrica aproximada
		x={(-0.5cm,-0.3cm)},
		y={(1cm,0cm)},
		z={(0cm,1cm)},
		line cap=round,
		line join=round
		]
		
		%------------------------------------------------
		% 1) Definir las placas: piso (z=0) y techo (z=2)
		%------------------------------------------------
		% Piso: rectángulo de (0,0,0) a (3,2,0)
		\coordinate (B1) at (0,0,0);
		\coordinate (B2) at (3,0,0);
		\coordinate (B3) at (3,2,0);
		\coordinate (B4) at (0,2,0);
		% Techo: rectángulo de (0,0,2) a (3,2,2)
		\coordinate (T1) at (0,0,2);
		\coordinate (T2) at (3,0,2);
		\coordinate (T3) at (3,2,2);
		\coordinate (T4) at (0,2,2);
		
		\fill[gray!50,draw=black] (B1) -- (B2) -- (B3) -- (B4) -- cycle;
		\fill[gray!50,draw=black] (T1) -- (T2) -- (T3) -- (T4) -- cycle;
		
		%------------------------------------------------
		% 2) Calcular centros y desplazar los postes hacia el interior
		%------------------------------------------------
		\coordinate (Bcenter) at (1.5,1,0);
		\coordinate (Tcenter) at (1.5,1,2);
		\def\alpha{0.2} % Factor de desplazamiento (0: en el vértice; 1: en el centro)
		
		\coordinate (B1a) at ($(B1)! \alpha ! (Bcenter)$);
		\coordinate (B2a) at ($(B2)! \alpha ! (Bcenter)$);
		\coordinate (B3a) at ($(B3)! \alpha ! (Bcenter)$);
		\coordinate (B4a) at ($(B4)! \alpha ! (Bcenter)$);
		\coordinate (T1a) at ($(T1)! \alpha ! (Tcenter)$);
		\coordinate (T2a) at ($(T2)! \alpha ! (Tcenter)$);
		\coordinate (T3a) at ($(T3)! \alpha ! (Tcenter)$);
		\coordinate (T4a) at ($(T4)! \alpha ! (Tcenter)$);
		
		% Dibujar los postes (varillas verticales) en gris
		%\draw[line width=2pt, line cap=round, gray] (B1a) -- (T1a);
		\draw[line width=2pt, line cap=round, gray!40,
		preaction={draw, line width=3pt, black}]
		(B1a) -- (T1a);
		%\draw[line width=2pt, line cap=round, gray] (B2a) -- (T2a);
		\draw[line width=2pt, line cap=round, gray!40,
		preaction={draw, line width=3pt, black}]
		(B2a) -- (T2a);
		%\draw[line width=2pt, line cap=round, gray] (B3a) -- (T3a);
		\draw[line width=2pt, line cap=round, gray!40,
		preaction={draw, line width=3pt, black}]
		(B3a) -- (T3a);
		%\draw[line width=2pt, line cap=round, gray] (B4a) -- (T4a);
		\draw[line width=2pt, line cap=round, gray!40,
		preaction={draw, line width=3pt, black}]
		(B4a) -- (T4a);
		
		%------------------------------------------------
		% 3) Definir los puntos de conexión en cada poste (y_i) con alturas distintas
		%------------------------------------------------
		\pgfmathsetmacro{\fYone}{0.70}  % En el poste del vértice 1
		\pgfmathsetmacro{\fYtwo}{0.60}  % En el poste del vértice 2
		\pgfmathsetmacro{\fYthree}{0.65} % En el poste del vértice 3
		\pgfmathsetmacro{\fYfour}{0.55}  % En el poste del vértice 4
		
		\coordinate (Y1) at ($(T1a)!\fYone!(B1a)$);
		\coordinate (Y2) at ($(T2a)!\fYtwo!(B2a)$);
		\coordinate (Y3) at ($(T3a)!\fYthree!(B3a)$);
		\coordinate (Y4) at ($(T4a)!\fYfour!(B4a)$);
		
		% Dibujar los puntos negros con etiquetas
		\fill (Y1) circle (0.8pt)
		node[above=3pt,left] {$y_1$};
		
		\fill (Y2) circle (0.8pt)
		node[left] {$y_2$};
		
		% y3: desplazar etiqueta hacia "afuera" y arriba
		\fill (Y3) circle (0.8pt)
		node at ($(Y3)+(0.58,0.14,0.05)$) {$y_3$};
		
		% y4: desplazar etiqueta hacia el frente-izquierda y arriba
		\fill (Y4) circle (0.8pt)
		node at ($(Y4)+(-0.12,0.12,0)$) {$y_4$};
		
		%------------------------------------------------
		% 4) Conectar los puntos y_i con varillas verdes (perímetro irregular)
		%------------------------------------------------
		%Original
		%\draw[very thick,teal] (Y1) -- (Y2) -- (Y3) -- (Y4) -- cycle;
		\draw[thick, double=teal!50] (Y1) -- (Y2) -- (Y3) -- (Y4) -- cycle;
		
		%------------------------------------------------
		% 5) Colocar las masas (bolas rojas) en los puntos medios de cada varilla verde
		%------------------------------------------------
		\coordinate (M1) at ($ (Y1)!0.5!(Y2) $);
		\coordinate (M2) at ($ (Y2)!0.5!(Y3) $);
		\coordinate (M3) at ($ (Y3)!0.5!(Y4) $);
		\coordinate (M4) at ($ (Y4)!0.5!(Y1) $);
		\shade[ball color=red] (M1) circle (3pt);
		\shade[ball color=red] (M2) circle (3pt);
		\shade[ball color=red] (M3) circle (3pt);
		\shade[ball color=red] (M4) circle (3pt);
		
		%------------------------------------------------
		% 6) Conectar cada masa con un resorte que va al techo
		%     Se conecta desde la masa hasta la proyección vertical en el techo,
		%     es decir, se sube  la diferencia entre z(M_i) y 2.
		%     Calculamos la coordenada z de cada masa como el promedio de las z de los y_i correspondientes.
		%------------------------------------------------
		\pgfmathsetmacro{\zMone}{((2-2*\fYone)+(2-2*\fYtwo))/2}
		\pgfmathsetmacro{\zMtwo}{((2-2*\fYtwo)+(2-2*\fYthree))/2}
		\pgfmathsetmacro{\zMthree}{((2-2*\fYthree)+(2-2*\fYfour))/2}
		\pgfmathsetmacro{\zMfour}{((2-2*\fYfour)+(2-2*\fYone))/2}
		
		\coordinate (M1_top) at ($(M1)+(0,0,{2-\zMone})$);
		\coordinate (M2_top) at ($(M2)+(0,0,{2-\zMtwo})$);
		\coordinate (M3_top) at ($(M3)+(0,0,{2-\zMthree})$);
		\coordinate (M4_top) at ($(M4)+(0,0,{2-\zMfour})$);
		
		\draw[blue, decorate, decoration={coil, segment length=5pt, amplitude=4pt}] (M1_top) -- (M1);
		\draw[blue, decorate, decoration={coil, segment length=5pt, amplitude=4pt}] (M2_top) -- (M2);
		\draw[blue, decorate, decoration={coil, segment length=5pt, amplitude=4pt}] (M3_top) -- (M3);
		\draw[blue, decorate, decoration={coil, segment length=5pt, amplitude=4pt}] (M4_top) -- (M4);
		
		% Opcional: marcar los puntos de conexión en el techo
		\fill (M1_top) circle (0.8pt);
		\fill (M2_top) circle (0.8pt);
		\fill (M3_top) circle (0.8pt);
		\fill (M4_top) circle (0.8pt);
		
	\end{tikzpicture}
	\caption{Variant of the mechanical system with emergent gauge symmetry. The absence of specific anchoring springs leads to a nontrivial kernel in the symplectic matrix, which the engine identifies as a translational gauge redundancy.}
	\label{fig:gauge_system}
\end{figure}

In such cases, the \texttt{MatrixStatus} returns \texttt{"Singular"}, and the engine preserves the final pre-symplectic matrix and its null vectors~\cite{Henneaux1992,GitmanTyutin1990}. This provides the researcher with the exact generators of the gauge transformations, facilitating the subsequent application of gauge-fixing conditions or the transition to a reduced phase space~\cite{Marsden1999,MarsdenWeinstein1974,ArnoldBook}. This transition from automated reduction to symmetry diagnostics ensures that the engine remains a reliable tool even when the symplectic structure is not globally invertible.

\subsection{Advantages of the association structure}

Encoding the physical model as an association offers three distinct advantages over functional definitions:
\begin{enumerate}
	\item \textbf{Modularity:} Constraints or external fields can be injected into the association as new keys without altering the core reduction code~\cite{Wolfram2002}.
	\item \textbf{Symbolic Hygiene:} By avoiding \texttt{q[t]}, we prevent the CAS from prematurely attempting to solve differential equations~\cite{Geddes1992,Fateman1991}, keeping the focus strictly on the algebraic topology of the constraints.
	\item \textbf{Traceability:} The input structure mirrors the mathematical definition of the Lagrangian $L = T - V$, ensuring that the computational object is isomorphic to the physical theory~\cite{ArnoldBook,Goldstein2002}.
\end{enumerate}

\begin{remark} From a computational standpoint, this immutable symbolic architecture is the operational counterpart to the geometric constraints established in Theorem \ref{theo1}. Unlike numerical solvers that may obscure singularities through rounding~\cite{Higham2002,Trefethen1997}, the association-based engine propagates symbolic dependencies exactly. This ensures that the structural regularization discussed in Section \ref{sec2} is preserved at the software level, maintaining full visibility of potential degeneracy loci and bifurcation conditions throughout the automated workflow~\cite{Kuznetsov2004,Seydel2010}. \end{remark}

This architectural decision transforms the Faddeev-Jackiw reduction from a manual calculation aid into a robust \textit{computational engine} capable of handling high-dimensional singular systems with minimal user intervention.

\section{Conclusion and future outlook}\label{6}

In this work, we have established a rigorous equivalence between the Faddeev-Jackiw symplectic reduction and the Matrix Bordering Technique. By proving that the iterative extension of the phase space is a geometrically constrained instance of matrix bordering (Theorem \ref{theo1}), we have provided a solid algebraic foundation for the automation of constrained dynamics~\cite{GolubVanLoan2013,Galantai2001}. This theoretical insight allowed us to construct a symbolic engine in \textit{Mathematica} that successfully regularizes singular Lagrangians without relying on heuristic simplifications often introduced in procedural implementations~\cite{Wolfram2002,Geddes1992}.

Our validation on representative singular systems~\cite{HojmanUrrutia1981,Brown2023} demonstrates that the engine correctly handles subtle dependencies in the constraint algebra, recovering the physical symplectic manifold $\mathcal{M}^*$ and preserving the parametric structures essential for stability analysis~\cite{Kuznetsov2004,Guckenheimer1983}.

%	\subsection*{Towards Field Theories and Tensor Calculus}

%	While the current implementation focuses on finite-dimensional systems (point mechanics), the algebraic architecture presented herein serves as the prototypical kernel for a more ambitious extension towards field theories. The matrix bordering logic established in Section \ref{sec2} is structurally agnostic to the dimensionality of the phase space. In the transition to continuum mechanics or gauge theories, the symplectic matrices generalize to integro-differential operators, and the partial derivatives lift to variational derivatives.

%	Consequently, this work lays the groundwork for a generalized "Tensor Faddeev--Jackiw Engine." Future developments will aim to integrate this bordering logic with symbolic tensor calculus packages (such as \textit{xAct}~\cite{mgarcia} or native tensorial tools). Such an extension would allow the symbolic automaton to handle systems with infinitely many degrees of freedom—such as Maxwell or Yang–Mills theories—by treating indices abstractly and performing the bordering operations directly on the operator algebra. In this perspective, constraint handling emerges from the algebraic structure of the symplectic form itself, rather than from a case-by-case classification. 	We argue that the algorithmic robustness demonstrated here for finite-dimensional systems constitutes a necessary prerequisite for navigating the infinite-dimensional constraint surfaces of relativistic field theories, where structural transparency and symbolic control are essential.

\subsection*{Towards Field Theories and Tensor Calculus}

While the current implementation focuses on finite-dimensional systems (point mechanics), the algebraic architecture developed in this work is conceived as a prototypical kernel for future extensions toward field theories. In particular, the matrix bordering logic established in Section~\ref{sec2} does not rely on any finite-dimensional assumption per se, but rather on structural properties of the underlying symplectic form and its associated constraint operators~\cite{ArnoldBook,Marsden1999}.

In the transition from point mechanics to continuum systems or gauge field theories~\cite{Witten1989,Henneaux1992}, the finite-dimensional symplectic matrices are naturally replaced by integro-differential operators acting on appropriate functional spaces, while partial derivatives lift to variational derivatives~\cite{Olver1993,Saunders1989}. At this level, the relevant mathematical setting is no longer matrix algebra but operator theory on Hilbert or Banach spaces. It is therefore important to emphasize that operator-level analogues of the Schur complement and matrix bordering constructions have been systematically studied in the functional-analytic literature, under suitable hypotheses of boundedness, complementability, and domain compatibility \cite{Bacuta2009,DouglasFillmore1974}. This body of work provides a rigorous analytical backdrop for exploring infinite-dimensional extensions of the present algebraic framework.

From this perspective, the present contribution may be viewed as establishing the finite-dimensional backbone of a prospective \emph{Tensor Faddeev-Jackiw Engine}. Future developments will aim to integrate the current bordering logic with symbolic tensor calculus frameworks, such as \textit{xAct}~\cite{mgarcia} or native tensorial tools~\cite{Nakahara2003}, in order to treat indices abstractly and to perform the bordering operations directly at the level of operator-valued symplectic forms. Such an extension would enable the systematic symbolic handling of constrained field theories, including Maxwell and Yang-Mills models~\cite{Witten1989,Schwarz1993}, while preserving the structural transparency that is essential in relativistic contexts.

Finally, we stress that the algorithmic robustness demonstrated here in the finite-dimensional setting constitutes a necessary prerequisite for addressing infinite-dimensional constraint surfaces. In field-theoretic applications, where gauge symmetries~\cite{Henneaux1992,Rothe1997}, functional degeneracies, and bifurcation phenomena may coexist, maintaining symbolic control and structural clarity is indispensable~\cite{Kuznetsov2004,Seydel2010}. The present work therefore does not claim to resolve the infinite-dimensional case, but rather to provide a mathematically coherent and computationally disciplined foundation upon which such extensions can be built.

\subsection*{Software availability}

The symbolic engine developed in this work, \texttt{BorderedFJReduction}, is implemented as a Wolfram Language package. It provides a complete, self-contained realization of the matrix bordering formulation of the Faddeev-Jackiw algorithm, together with a structured and queryable interface for constructing and analyzing the resulting symplectic data.

The package is distributed as a standard Wolfram Language paclet and is publicly available through the project's GitHub repository:

\begin{center}
	\url{https://github.com/echanlopez/BorderedFJReduction}
\end{center}

The repository contains the complete package source code, documentation notebooks, validation examples, and supporting material for reproducing the computational results presented in this work. The implementation is designed to preserve the symbolic and parametric structure of the reduction throughout the computation, allowing the resulting framework to be used both as a computational engine and as a tool for symbolic exploration of constrained systems.

The software is released under the MIT License. For long-term preservation, versioned archiving, and persistent citation, the \texttt{BorderedFJReduction} v0.1.2 release is deposited in Zenodo and is assigned the following persistent DOI:

\begin{center}
	\url{https://doi.org/10.5281/zenodo.18487436}
\end{center}

The Zenodo record provides a citable, version-specific archive of the software corresponding to the release described in this work.

\subsection*{Appendix A. Determinantal factorization and kernel--Poisson identity: proofs}

This appendix provides the technical backbone of Theorem~\ref{theo1}. We prove a determinantal factorization lemma for antisymmetric bordered matrices with singular anchor block (Lemma~\ref{lem:rank}) and an exact coefficient-by-coefficient identification of the reduced constraint matrix with the Hessian of the symplectic potential restricted to $\ker(f^{(0)})$---equivalently, the Faddeev--Jackiw bracket matrix of the generated constraints (Proposition~\ref{prop:schurpoisson}).

\medskip
\noindent\textbf{Notation.}
For a real antisymmetric matrix $A \in \mathbb{R}^{n\times n}$ of rank $r$, there exists an orthogonal matrix $Q$ such that $Q^{\top} A Q = \mathrm{diag}(J, 0)$, where $J \in \mathbb{R}^{r\times r}$ is block-diagonal with $r/2$ invertible $2\times 2$ symplectic blocks $\mu_\ell \begin{psmallmatrix} 0 & 1 \\ -1 & 0 \end{psmallmatrix}$, $\mu_\ell > 0$. The last $n-r$ columns of $Q$ span $\ker(A)$ and furnish the orthonormal basis $N$ used throughout.

\begin{lemma}[Determinantal factorization for antisymmetric bordered matrices]\label{lem:rank}
	Let $f^{(0)} \in \mathbb{R}^{n \times n}$ be real antisymmetric with $\mathrm{rank}(f^{(0)}) = r = n - d$, $d = \dim\ker(f^{(0)}) \geq 1$, and let $B \in \mathbb{R}^{n \times k}$. Let $N \in \mathbb{R}^{n\times d}$ have orthonormal columns spanning $\ker(f^{(0)})$, and define the \emph{reduced constraint matrix}
	\[
	\Gamma \;\equiv\; N^{\top} B \;\in\; \mathbb{R}^{d\times k}.
	\]
	Then the extended symplectic matrix
	\[
	f^{(m)} \;=\; \begin{pmatrix} f^{(0)} & B \\ - B^{\top} & 0 \end{pmatrix}
	\]
	satisfies $\det(f^{(m)}) = \det(J)\,\det(R)$, where $J$ is the nonsingular block of $f^{(0)}$ and $R=\begin{psmallmatrix} 0 & \Gamma \\ -\Gamma^{\top} & M\end{psmallmatrix}$ with $M$ antisymmetric. In particular, if $k = d$ then
	\begin{equation}\label{eq:rank-lemma}
		\det\bigl(f^{(m)}\bigr) \;=\; \Bigl(\textstyle\prod_{\ell=1}^{r/2}\mu_\ell^{2}\Bigr)\,\det(\Gamma)^2 ,
	\end{equation}
	where $\mu_\ell>0$ are the Pfaffian factors of $f^{(0)}$. Consequently $f^{(m)}$ is nonsingular (i.e.\ of rank $n+k$) if and only if:
	\begin{enumerate}
		\item[\textup{(i)}] $k = d = n - r$, and
		\item[\textup{(ii)}] $\det(\Gamma) \neq 0$, i.e.\ $\mathrm{rank}(\Gamma) = d$,
	\end{enumerate}
	that is, if and only if the number of independent generated constraints equals $\dim\ker(f^{(0)})$ and their gradients pair non-degenerately with the null space of $f^{(0)}$.
\end{lemma}

\begin{proof}
	Since $f^{(0)}$ is real antisymmetric of rank $r$, there exists an orthogonal $Q \in \mathbb{R}^{n\times n}$ such that
	\[
	Q^{\top} f^{(0)} Q \;=\; \begin{pmatrix} J & 0 \\ 0 & 0 \end{pmatrix}, \qquad J = \bigoplus_{\ell=1}^{r/2}\mu_\ell\!\begin{psmallmatrix}0 & 1 \\ -1 & 0\end{psmallmatrix}\ \text{invertible},\quad \mu_\ell>0,
	\]
	and we may take the last $d$ columns of $Q$ to be the orthonormal kernel basis $N$. Partition $Q^{\top} B = (B_u^{\top}, B_w^{\top})^{\top}$ with $B_u \in \mathbb{R}^{r\times k}$ and $B_w \in \mathbb{R}^{d\times k}$; by construction $B_w = N^{\top}B = \Gamma$.
	
	Conjugating $f^{(m)}$ by $\mathrm{diag}(Q, I_k)$ yields
	\[
	\widetilde{f}^{(m)} \;=\; \begin{pmatrix} J & 0 & B_u \\ 0 & 0 & \Gamma \\ -B_u^{\top} & -\Gamma^{\top} & 0 \end{pmatrix}.
	\]
	Because $J$ is invertible, the congruence with the unit upper-triangular
	$T = \begin{psmallmatrix} I_r & 0 & -J^{-1}B_u \\ 0 & I_{d} & 0 \\ 0 & 0 & I_k \end{psmallmatrix}$
	(with $\det T = 1$) eliminates the $(1,3)$ and $(3,1)$ blocks:
	\[
	T^{\top}\, \widetilde{f}^{(m)}\, T \;=\; \begin{pmatrix} J & 0 & 0 \\ 0 & 0 & \Gamma \\ 0 & -\Gamma^{\top} & M \end{pmatrix},
	\qquad M = B_u^{\top} J^{-1} B_u\ \text{(antisymmetric)}.
	\]
	Since congruence preserves the determinant up to $(\det T)^2 = 1$, we obtain $\det(f^{(m)}) = \det(J)\,\det(R)$ with $R = \begin{psmallmatrix} 0 & \Gamma \\ -\Gamma^{\top} & M \end{psmallmatrix}$ and $\det(J) = \prod_{\ell}\mu_\ell^{2}$.
	
	It remains to evaluate $\det(R)$ when $k=d$. Then $R$ is antisymmetric of even size $2d$, so $\det(R)=\mathrm{Pf}(R)^2$. The block congruence with $\begin{psmallmatrix} I_d & 0 \\ X & I_d \end{psmallmatrix}$, where $X$ is chosen so that $X\Gamma - (X\Gamma)^{\top} = M$ (solvable because $M$ is antisymmetric and $\Gamma$ invertible), maps $R$ to $\begin{psmallmatrix} 0 & \Gamma \\ -\Gamma^{\top} & 0 \end{psmallmatrix}$ without changing its Pfaffian; hence $\mathrm{Pf}(R) = (-1)^{d(d-1)/2}\det(\Gamma)$ and $\det(R)=\det(\Gamma)^2$, independently of the image block $M$. This proves~\eqref{eq:rank-lemma}. Since $\prod_\ell\mu_\ell^2>0$, $f^{(m)}$ is nonsingular iff $\det(\Gamma)\neq 0$; and $\Gamma\in\mathbb{R}^{d\times k}$ can have $\det\neq 0$ only when it is square, i.e.\ $k=d$, giving conditions (i) and (ii). If $k\neq d$ the size $n+k$ has the parity of $n-r+k$ and $\Gamma$ is not square, so $f^{(m)}$ is singular.
\end{proof}

\begin{remark}
	Condition (ii) of Lemma~\ref{lem:rank} is the precise meaning of the informal phrase that the bordering vectors be ``structurally compatible with the kernel of $f^{(0)}$''. It is the matrix-theoretic translation of the physical requirement that the consistency constraints genuinely restrict the dynamics on the singular locus: the constraint gradients, paired with the null vectors of $f^{(0)}$, must form a nondegenerate $d\times d$ matrix. The reduced constraint matrix $\Gamma = N^{\top}B$ isolates exactly this interaction between the bordering block and the null space of $f^{(0)}$; since the singular directions are the only ones that survive the elimination of the nonsingular symplectic core, $\Gamma$ contains precisely the information that governs the regularity of the bordered matrix.
\end{remark}

\medskip
\noindent\textbf{The reduced constraint matrix as the effective coupling matrix.}
The determinantal factorization~\eqref{eq:rank-lemma} exhibits $\Gamma$ as the effective coupling matrix that survives once the \emph{nonsingular} block $J$ of $f^{(0)}$ is eliminated: the congruence of Lemma~\ref{lem:rank} removes the invertible symplectic core and leaves the residual antisymmetric pairing $R=\begin{psmallmatrix}0 & \Gamma \\ -\Gamma^{\top} & M\end{psmallmatrix}$, whose determinant, when $k=d$, equals $\det(\Gamma)^2$ regardless of the image correction $M$. Thus
\begin{equation}\label{eq:genschur-appx}
	\det(f^{(m)}) \neq 0 \;\iff\; \det(\Gamma) \neq 0, \qquad \Gamma = N^{\top}B \in \mathbb{R}^{d\times d}\ (k=d),
\end{equation}
so the regularity of the extended symplectic matrix is fully controlled by the nonsingularity of the kernel pairing $\Gamma$: it is the reduced object whose nonsingularity decides whether the bordering restores a nondegenerate symplectic form.

\begin{proposition}[Kernel--Poisson identity]\label{prop:schurpoisson}
	Let $\Omega_\alpha$, $\alpha = 1, \ldots, d$, be the independent constraints generated by the Barcelos-Neto--Wotzasek consistency condition, so that (up to independent recombination) $\Omega_\alpha = N_\alpha^{\,i}\,\partial_i V$ with $N_\alpha$ the $\alpha$-th null vector of $f^{(0)}$, and let $B_{j\alpha} = \partial_j \Omega_\alpha$. Let $\Gamma = N^{\top}B$ be the reduced constraint matrix of Lemma~\ref{lem:rank}. Then
	\begin{equation}\label{eq:schurpoisson-appx}
		\Gamma_{\alpha\beta} \;=\; \sum_{i,j=1}^{n} N_\alpha^{\,i}\,\bigl(\partial_i\partial_j V\bigr)\,N_\beta^{\,j} \;=\; \{\Omega_\alpha, \Omega_\beta\}_{\mathrm{FJ}}
	\end{equation}
	as an equality of matrices, coefficient by coefficient, where $\{\Omega_\alpha,\Omega_\beta\}_{\mathrm{FJ}}$ denotes the Hessian of the symplectic potential restricted to $\ker(f^{(0)})$. In particular the reduced matrix $\Gamma$ is \emph{literally} the constraint matrix of the theory, not merely isomorphic or equivalent to it modulo terms vanishing on the constraint surface; and its nondegeneracy is the second-class condition in the Dirac--Bergmann sense.
\end{proposition}

\begin{proof}
	By $\Gamma = N^{\top}B$ and $B_{j\alpha}=\partial_j\Omega_\alpha$,
	\[
	\Gamma_{\alpha\beta} \;=\; \sum_{i} N_\alpha^{\,i}\,\partial_i\Omega_\beta \;=\; \sum_{i} N_\alpha^{\,i}\,\partial_i\bigl(N_\beta^{\,j}\partial_j V\bigr) \;=\; \sum_{i,j} N_\alpha^{\,i}\,\bigl(\partial_i\partial_j V\bigr)\,N_\beta^{\,j},
	\]
	using that the null vectors $N_\beta$ are constant on the singular locus. This is exactly the Hessian of $V$ contracted with the kernel basis, which is by definition $\{\Omega_\alpha,\Omega_\beta\}_{\mathrm{FJ}}$. The identity is exact and requires no approximation; it holds off-shell as a pointwise matrix equality on the phase space, not merely modulo the constraint ideal.
\end{proof}

\begin{remark}
	In the standard Dirac--Bergmann classification, the constraint matrix $\mathcal{C}_{\alpha\beta} = \{\Omega_\alpha, \Omega_\beta\}$ whose nondegeneracy defines a second-class system~\cite{Dirac1964,Henneaux1992,GitmanTyutin1990} is, for constraints generated by consistency, the Hessian of the symplectic potential restricted to $\ker(f^{(0)})$. Proposition~\ref{prop:schurpoisson} establishes that this matrix coincides exactly with the reduced pairing $\Gamma = N^{\top}B$, and Lemma~\ref{lem:rank} shows that $\Gamma$ controls the regularity of $f^{(m)}$ through the factorization $\det(f^{(m)}) = (\prod_\ell\mu_\ell^2)\det(\Gamma)^2$. Consequently, the termination of the Faddeev--Jackiw algorithm in a nondegenerate symplectic form is equivalent to a purely algebraic second-class condition:
	\[
	\det(f^{(m)}) \neq 0 \;\iff\; \det\{\Omega_\alpha, \Omega_\beta\}_{\mathrm{FJ}} \neq 0.
	\]
	The regularization mechanism is thus entirely structural; it arises from the nondegenerate pairing of the bordering block with $\ker(f^{(0)})$ rather than from any numerical inversion procedure.
\end{remark}

\section*{Acknowledgements}

E.C.-L. and J.M.C. were supported by the SECIHTI under the program ``Estancias Posdoctorales por México'' with CVU numbers 422090 and 376886, respectively. A.M.-R. acknowledges financial support by UNAM-PAPIIT project No.~IG100224, UNAM-PAPIME project No.~PE109226, by SECIHTI project No.~CBF-2025-I-1862 and by the Marcos Moshinsky Foundation.

\bibliographystyle{unsrt}

\end{document}